\documentclass[reqno,a4paper,11pt]{amsart}
\usepackage{a4wide}
\usepackage{times}
\usepackage{abstract}
\usepackage{graphics}
\usepackage{amsthm,amsmath,amssymb,amsfonts,mathrsfs}

\theoremstyle{plain}
 
\newtheorem{theorem}{Theorem}

\newtheorem{lemma}[theorem]{Lemma}

\theoremstyle{definition}

\theoremstyle{remark}

\renewcommand{\leq}{\leqslant}

\renewcommand{\ge}{\geqslant}

\newcommand{\A}{\mathcal{A}}
\newcommand{\Z}{\mathbb{Z}}
\newcommand{\Zx}{\mathbb{Z}^{\times}}

\newcommand{\setup}{\texttt{Setup}}
\newcommand{\extract}{\texttt{Extract}}
\newcommand{\key}{\texttt{KeyExchange}}
\newcommand{\mpk}{\mathsf{mpk}}
\newcommand{\sk}{\mathsf{sk}}
\newcommand{\msk}{\mathsf{msk}}
\newcommand{\skid}{\mathsf{sk}_{\id}}
\newcommand{\skidp}[1]{\mathsf{sk}_{\id_{#1}}} 
\newcommand{\id}{\mathsf{id}}

\newcommand{\ssk}{\mathsf{ssk}}
\newcommand{\hash}{\mathcal{H}}
\newcommand{\QR}{{\textsf{{{QR}}}}}
\newcommand{\J}{\textsf{{J}}}
\newcommand{\legendre}[2]{\left(\dfrac{#1}{#2}\right)}

\newcommand{\Tun}{\mu^{\alpha_1}R_1}
\newcommand{\Tde}{\mu^{\alpha_2}R_2}

\title{Cryptanalysis of an Identity-Based Authenticated Key Exchange Protocol}

\author{Younes Hatri}\thanks{\texttt{hatri.younes@hotmail.fr}. Universit\'e des Sciences et de la Technologie Houari Boumediene, Bab Ezzouar 16111,  Algeria.}

\author{Ayoub Otmani} 
\thanks{\texttt{ayoub.otmani@univ-rouen.fr}. Normandie Univ, France; UR, LITIS, F-76821 Mont-Saint-Aignan, France.}

\author{Kenza Guenda}
\thanks{\texttt{ken.guenda@gmail.fr}. Universit\'e des Sciences et de la Technologie Houari Boumediene, Bab Ezzouar 16111,  Algeria.}

\begin{document}

\maketitle

\begin{abstract}
Authenticated Key Exchange (AKE) protocols represent an important cryptographic mechanism that enables several parties to communicate securely over an open network. Elashry, Mu and Susilo proposed an Identity Based Authenticated Key Exchange (IBAKE) protocol where different parties establish secure communication by means of their public identities.The authors also introduced a new security notion for IBAKE protocols called resiliency, that is, if the secret shared key is compromised, the entities can generate another shared secret key without establishing a new session between them. They then claimed that their IBAKE protocol satisfies this security notion.

We analyze the security of their protocol and prove that it has a major security flaw which renders it insecure against an impersonation attack.  We also disprove the resiliency property of their scheme by proposing an attack where an adversary can compute any share secret key if just one secret bit is leaked.
\end{abstract}


\section{Introduction}

Key agreement protocols permit different parties
to share  a common secret key which in turn can be used for 
different cryptographic goals like communication encryption, data
integrity, \textit{etc}.
The first practical solution to the problem of key-distribution  
is the famous Diffie-Hellman protocol \cite{DH76}.
However, it does not prevent from Man-In-The-Middle attacks because it
does not authenticate the involved parties.
A key agreement protocol provides key authentication
if each entity involved in the exchange is assured that no other entity
can learn the shared secret key. 
There exist several methods to broadcast authenticated keys. Classically it requires
public-key certificates with public key infrastructures. Another very interesting approach is to use
public data like identities  to generate authenticated keys. 

The idea of using identities in cryptography dates back to Shamir's paper \cite{S85} 
where he asks how to achieve a public key encryption 
scheme that allows to compute public keys from arbitrary strings like user's identity 
(an email, phone number, \emph{etc}). 
Consequently, electronic certificates are no more required and more importantly
it eliminates the need for large-scale public key infrastructure.
Although Shamir introduced in \cite{S85} the concept of Identity-Based Encryption (IBE), 
he was not able to propose one.
The construction remained an open problem until
Boneh and Franklin \cite{BF01} and Cocks \cite{C01} proposed IBE schemes in 2001. 
The Boneh-Franklin scheme \cite{BF01} makes use of bilinear maps which then sparked
a lot of works \cite{BX04,W05}. Recently,
lattices have also been used in the design of IBE schemes \cite{GPV08} which
gave rise to a large number of schemes.

Cocks builds in \cite{C01} an  IBE scheme based on the
quadratic residuosity  problem modulo an RSA integer.
It is time-efficient compared with pairing-based IBE systems, but unfortunately
ciphertexts are very long. Boneh, Gentry and Hamburg (BGH) solved
the problem of Cocks' scheme by presenting a space-efficient scheme without pairings 
but at the cost of a less time-efficient scheme \cite{BGH07}.

The concept of IBE was extended to authenticated key exchange (AKE) protocols.
Smart \cite{S01} presented a two-pass Identity-Based AKE (IBAKE) using Weil pairings and
merging the ideas of Boneh and Franklin \cite{BF01} with tripartite Diffie-Hellman (DH) protocol of Joux \cite{J04}.
This work was then followed by several works.
Recently, Elashry, Mu and Susilo \cite{EMS15} proposed another IBAKE protocol
and introduced a new security notion called \emph{resiliency}. 
A key exchange protocol is said to be \emph{resilient}
when parties are able to generate new shared secret keys 
without establishing a new session between them, even if 
a secret shared key has been compromised. 
The IBAKE protocol proposed by \cite{EMS15} builds upon  the IBE encryption scheme of \cite{BGH07} 
and it is claimed in \cite{EMS15} to be resilient.

\subsection{Our contribution.}

In this paper, we analyze the security of Elashry, Mu and Sussilo (EMS) protocol \cite{EMS15} and
prove that it has a major security flaw
which renders it insecure against an impersonation attack.  We are indeed able
to prove that the protocol is insecure against a very simple man-in-the-middle attack.

\medskip

We also disprove the resiliency property of the EMS scheme by proposing
an attack where an adversary can compute in time quartic in the security parameter 
the secret shared key from the knowledge of a single secret bit. Our method is similar to the one given 
in \cite{TITN16} to attack an IBE encryption scheme proposed in \cite{JB09}.

\medskip

The rest of this paper is organized as follows. In Section~\ref{sec:prelim}, 
we recall the definition and notion for IBAKE protocols.
In Section~\ref{sec:EMS}, we present Elashry, Mu, Sussilo (EMS) IBAKE protocol \cite{EMS15}.
In Section~\ref{sec:attack}, we describe our attacks against this protocol. 
In Section~\ref{sec:discussion}, we discuss the question of repairing the scheme.
Finally, in Section~\ref{sec:conclusion}  we conclude the paper.

\section{Preliminaries} \label{sec:prelim}
\subsection{IBAKE Protocol}

We shall assume that a trusted authority is responsible for the creation and
distribution of users' private keys. An Identity-Based Authenticated Key Exchange protocol (IBAKE) \cite{S01,MB05} is
defined by three algorithms: $\setup()$, $\extract()$ and $\key()$.

\begin{enumerate}
  \item $(\msk,\mpk) \leftarrow \setup(\lambda)$. The authority takes as input a security parameter $\lambda$ 
and generates public parameters that are denoted by $\mpk$ and a \emph{master secret key} $\msk$.

  \item $\skid \leftarrow \extract(\msk,\id)$. Given an identity $\id$, the authority uses his master key
$\msk$ to generate the private key $\skid$ corresponding to $\id$.

  \item $\ssk \leftarrow \key(\id_1,\id_2).$ Two parties $P_1$ and $P_2$ with system parameters
$(\id_1,\skidp{1})$ and $(\id_2,\skidp{2})$ respectively generate a
shared secret key $\ssk$.
\end{enumerate}

\subsection{Quadratic Residues and Jacobi Symbol}
For any integer $N \ge 2$ we denote by $\Zx_N$ the multiplicative group of integers modulo $N$.
Let $y \in \Zx_N$ then we say that $y$ is a \emph{quadratic residue} in
$\Zx_N$ if there exists $x \in \Zx_N$ such that:
\[
 y\equiv x^2 \mod N.
\]
The set of quadratic residues in $\Zx_N$ is denoted by $\QR(N)$:
\[
\QR(N) = \left \{  y \in \Zx_N ~:~ \exists x \in \Zx_N, ~ y = x^2 \mod N \right \}.
\]
Let $p$ be an odd prime number,
we define the Legendre symbol  of $x \in \Z$ with respect to $p$ as
\begin{equation*}
\legendre{x}{p}
=
x^\frac{p-1}{2} \mod p.
\end{equation*}
We recall that $\legendre{x}{p}$ belongs to $\{-1,0,1\}$ and enables to determine if $x$
is a quadratic residue since we have:
\[
\legendre{x}{p} = \left \{ \begin{array}{rcl}
1 & \text{ if}  & x \in \QR(N) \\
-1 & \text{ if}  & x \notin \QR(N) \text{ and } x \neq 0 \mod p\\
0 & \text{ if}  & x = 0 \mod p. \\
\end{array}
\right.
\]
The Legendre symbol is extended to any
odd positive integer $N = p_1^{\alpha_1}\cdots p_k^{\alpha_k}$  where $p_1,\dots{},p_k$
are pairwise different prime numbers
and $\alpha_1,\dots,\alpha_k$ are positive integers.
This generalization is called the Jacobi symbol and is defined as:
\begin{equation*}
\left(\dfrac{x}{N}\right)
=
\left(\dfrac{x}{p_1}\right)^{\alpha_1} \cdots \left(\dfrac{x}{p_k}\right)^{\alpha_k}.
\end{equation*}
The subset of $\Z_N$ with symbol to equal $1$ is denoted by $\J(N)$. Note that $\QR(N)$ is a subset of $\J(N)$.

The \emph{quadratic residuosity assumption}  states that, for any integer $N =  pq$,
where $p$ and $q$ are different prime numbers that are picked at random,
there exists no probabilistic polynomial time  algorithm that is able to distinguish between the distribution
of samples drawn from $\QR(N)$ and the distribution
of samples picked from $\J(N) \setminus \QR(N)$ (see \cite{BGH07} for more details).

\subsection{Solving $R x^2 + S y^2 = 1 \mod N$.} \label{sec:SolvingQuad}

Assuming that $N = pq$ where $p$ and $q$ are different prime numbers, Boneh, Gentry and Hamburg presented in \cite{BGH07}
an efficient algorithm to solve in $\Z_N$ an equation of the form:
\begin{eqnarray} \label{eq:BGH}
R x^2 + S y^2 = 1 \mod N
\end{eqnarray}
where $R$ and $S$ are in $\Z_N$.
They considered the following ternary quadratic form over $\Z$:
\begin{eqnarray} \label{eq:ternary}
\widetilde{R} x^2 + \widetilde{S} y^2 - z^2 = 0.
\end{eqnarray}
with $\widetilde{R}$, $\widetilde{S}$ in $\Z$. A classical result of Legendre \cite{BGH07} says 
that \eqref{eq:ternary} has a solution $(x,y,z) \in \Z^3$ if 
there exist $\widetilde{r}$ and $\widetilde{s}$ in $\Z$ such that
\begin{eqnarray} \label{eq:rs}
\widetilde{R} = 
 \widetilde{r}^2 
\;\mod \widetilde{S}
\;\;\;\;\text{ and }\;\;\;\;
\widetilde{S} 
=
\widetilde{s}^2
\;\mod \;\widetilde{R}.
\end{eqnarray}
Cremona and Rusin proposed  in \cite{CR03} an algorithm using lattice reduction to solve \eqref{eq:ternary}
assuming that \eqref{eq:rs} holds. Furthermore, if  $\widetilde{R} = R \mod N$ and $\widetilde{S} = S \mod N$ 
then a solution to \eqref{eq:ternary} also gives a solution to \eqref{eq:BGH}.
Consequently, solving \eqref{eq:BGH} consists in finding prime numbers $\widetilde{R}$, $\widetilde{S}$ and integers 
$\widetilde{r}$, $\widetilde{s}$ 
such that
$\widetilde{R} = R \mod N$, $\widetilde{S} = S \mod N$ and $\widetilde{r}$, $\widetilde{s}$ 
satisfy \eqref{eq:rs}. There exist several possible candidates $(\widetilde{R},\widetilde{S})$ from a given couple $(R,S)$ 
but Boneh and Franklin proposed a \emph{deterministic} polynomial-time algorithm that finds a specific $(\widetilde{R},\widetilde{S})$ 
which leads to a solution to \eqref{eq:BGH}.
For more details we refer the reader to \cite{BGH07}. 

Finally we state an important lemma that shows an important property  used in \cite{BGH07} and \cite{EMS15}.

\begin{lemma} \label{lem:Jacobi_xy}
Assume that $R$ and $S$ belong to $\QR(N)$ and let $(x,y)$ be a solution
to \eqref{eq:BGH}. Then we have the following equality:
\[
\left(\frac{1 + x \sqrt{R}}{N}\right)
=
\left(\frac{2 + 2y \sqrt{S}}{N}\right).
\]
\end{lemma}

\begin{proof} We have in $\Z_N$ the following equality:
\begin{align*}
\left(x \sqrt{R} + 1 \right)
\left(2y\sqrt{S} +2 \right)
& = 2 x y \sqrt{R S} + 2 x \sqrt{R} + 2 y \sqrt{S} + 2 \\
& =  \left( x \sqrt{R} + y \sqrt{S}  + 1 \right)^2.
\end{align*}
The last equality is obtained by using \eqref{eq:BGH}.
\end{proof}

\section{Elashry-Mu-Susilo (EMS) IBAKE Scheme} \label{sec:EMS}

\begin{figure}
\begin{center}
\includegraphics{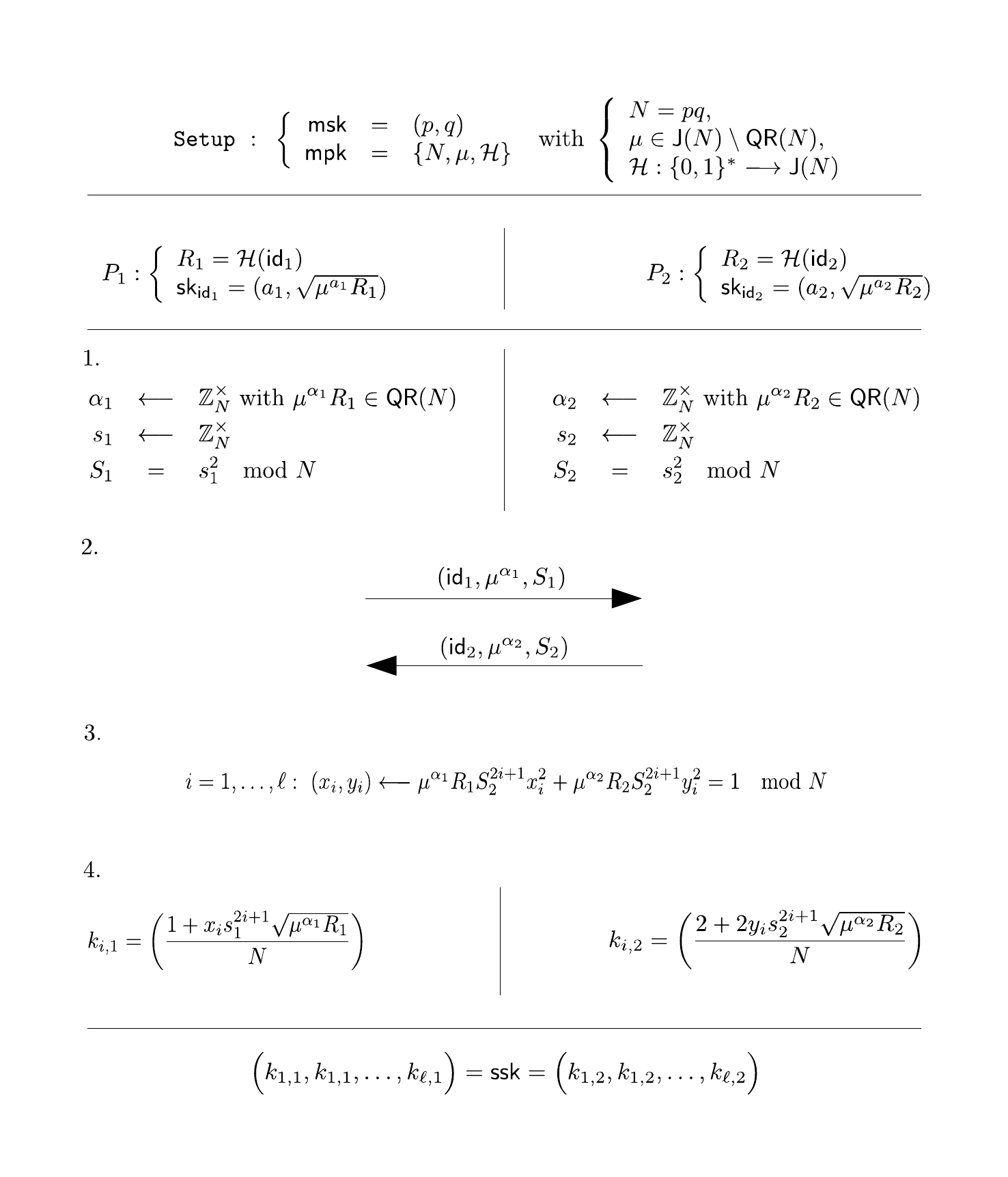}
\end{center}
\caption{Elashry-Mu-Susilo (EMS) IBAKE protocol.} \label{fig:ems}
\end{figure}

EMS scheme \cite{EMS15} is specified by the following algorithms (Fig. \ref{fig:ems}):

\begin{enumerate}
  \item $(\msk,\mpk) \leftarrow \setup(\lambda)$. 
The authority generates two prime numbers $p$ and $q$ according to security parameter $\lambda$.
It also picks $\mu \in \J(N)\setminus \QR(N)$ and  chooses a hash function
$\hash : \{0,1\}^* \longrightarrow \J(N)$.
The master public key is then $\mpk = \{N,\mu,\hash\}$ where $N=pq$
and the master secret key is $\msk = (p,q)$.

  \item $\skid \leftarrow \extract(\msk,\id)$. Given an identity $\id$, the authority
generates $R=\hash(\id)$. Since $R$ is in $\J(N)$ then either $R$ or
$\mu R$ belongs to $\QR(N)$. The authority chooses $a$ in $\{0,1\}$ such that
$\mu^a R$ belongs   to $\QR(N)$. It then picks at random one of the four possible square roots
of $\mu^a R$. We denote it by $\sqrt{\mu^a R}$.
The private key for identity $\id$ is then $\skid = (a,\sqrt{\mu^a R})$.

\item $\ssk \leftarrow \key(\id_1, \id_2)$. Party $P_1$ with identity $\id_1$
and system parameter $R_1 = \hash(\id_1), \sk_{\id_1}$
chooses two random values $s_1$ and $\alpha_1$ in $\Zx_N$ such that
\footnote{$P_1$ can easily find $\alpha_1 \in \Z$
such that $
\mu^{\alpha_1}R_1\in \QR$ and can even compute $\sqrt{\mu^{\alpha_1}R_1}$
from its private key $\sk_{\id_1} = (a_1,\sqrt{\mu^{a_1} R_1})$.
Indeed, $P_1$ chooses $\alpha_1$ to be equal to $2t+a_1$ for a random integer $t \in \Z$
so that 
\[
\mu^{\alpha_1}R_1 = \mu^{2t+a_1}R_1 =
\left(\mu^t \sqrt{\mu^a R_1}\right)^2.
\]}
\[
\mu^{\alpha_1}R_1\in \QR(N).
\]
$P_1$ then  sends
$(\id_1,\mu^{\alpha_1},S_1)$ to $P_2$ where $S_1 = s_1^2 \mod N$ and keeps secret
$(\alpha_1,s_1)$.
$P_2$  with identity $\id_2$
and system parameter $R_2 = \hash(\id_1), \sk_{\id_2}$ also performs the same procedure by
choosing two random values $s_2$ and $\alpha_2$ in $\Zx_N$ such that
\[
\mu^{\alpha_2}R_2  \in \QR(N).
\] 
Then $P_2$ sends $(\id_2,\mu^{\alpha_2},S_2)$ to $P_1$
with $S_2 =s_2^2 \mod N$, and keeps secret
$(\alpha_2,s_2)$.

\medskip

Each party $P_1$ and $P_2$ 
solves independently for each $i = 1,\dots,\ell$
the equation:
\begin{equation} \label{eq:EMS}
\Tun S_1^{2i+1} x_i^2
+
\Tde S_2^{2i+1} y_i^2
= 1 \mod N.
\end{equation}
From \emph{the} solution $(x_i,y_i)$ and its private key
$P_1$ is then able to compute the quantity $k_{i,1} \in \{ -1, 1\}$ where
\[
k_{i,1}= \left(\frac{1 + x_i  s_1^{2i+1} \sqrt{\Tun}}{N}\right).
\]
$P_2$ computes  $k_{i,2} \in \{ -1, 1\}$ from \emph{the} solution
$(x_i,y_i)$ and its private as the following:
 \[
  k_{i,2}=\left(\frac{2 + 2y_i s_2^{2i+1} \sqrt{\Tde} }{N}\right).
 \]
\end{enumerate}
By Lemma~\ref{lem:Jacobi_xy} we know that $k_{i,1} = k_{i,2}$ and therefore
the shared secret key $\ssk$ is 
\begin{equation*}
\Big(k_{1,1},k_{1,1},\dots, k_{\ell,1} \Big)
= 
\Big(k_{1,2}, k_{1,2}, \dots, k_{\ell,2} \Big)
=
\ssk.
\end{equation*}

\section{Cryptanalysis} \label{sec:attack}

\subsection{Impersonation Attack} 

\begin{figure}
\begin{center}
\includegraphics{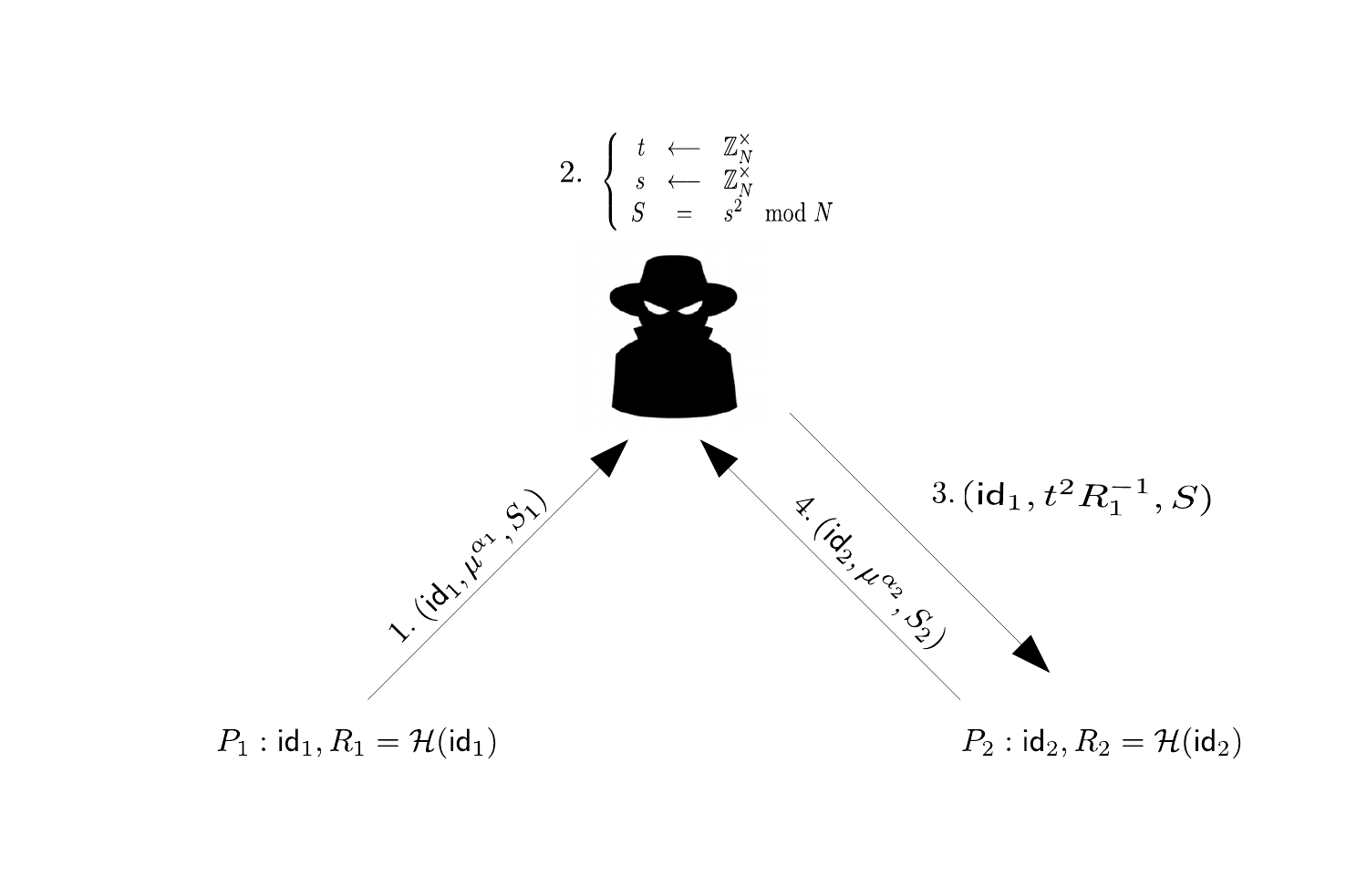}
\end{center}
\caption{Impersonation attack.} \label{fig:impersonation}
\end{figure}

The EMS protocol displays from its definition a major security flaw: it does not prevent from
parties to be impersonated by an adversary. 
The protocols does not ensure any authentication during the exchange. 
In the following we explain a simple man-in-the-middle attack. 
Let us assume that an adversary $\A$ receives and forwards data exchanged 
between $P_1$ and $P_2$ whose parameters are  respectively
$(R_1 = \hash(\id_1), \sk_{\id_1})$ and $(R_2 = \hash(\id_2), \sk_{\id_2})$. 
We will now show how $\A$ can easily impersonate $P_1$.

When $P_1$ sends its session identifier $(\id_1,\mu^{\alpha_1},S_1)$ to $P_2$,
$\A$ intercepts it and  chooses randomly $t$ and $s$ in $\Zx_N$, 
computes $S = s^2 \mod N$
then sends to $P_2$ the quantity
$\left(\id_1,\frac{t^2}{R_1},S \right)$.
$P_2$ also sends its session identifier $(\id_2,\mu^{\alpha_2},S_2)$ that is
intercepted by $\A$. Upon receiving $(\id_1,\frac{t^2}{R_1},S)$,
$P_2$ computes first $T = \frac{t^2}{R_1} R_1 \mod N$ which turns out to be 
$t^2 \mod N$.
Therefore $\A$ and $P_2$ have both to solve for $i = 1,\dots{},\ell$
the (common) equations:
\[
T S^{2i+1} x^2 + \Tde S_2^{2i+1} y^2 = 1 \mod N.
\]
Then $\A$ and $P_2$ share the same secret key $\ssk = (k_1,\dots{},k_\ell)$
since for any $i \ge 1$:
\[
k_{i}= \left(\frac{1 + x_i t s^{2i+1}}{N}\right)
=
\left(\frac{2 + 2 y_i \sqrt{\Tde} s_2^{2i+1}}{N}\right).
\]

The main reason why this attack is possible comes from the fact that each party in the
protocol perform computations without involving data that identify the correspondent.
Hence EMS protocol does not satisfy the basic property of authentication that 
any AKE protocol must satisfy.
In the next section, we analyze further the security of the protocol 
by showing that EMS  does not even ensure
the resiliency property \cite{EMS15}.

\subsection{Attack Against the Resiliency Property}

We assume that two parties $P_1$ and $P_2$ managed to share a secret key
$\ssk = (k_1,\dots{},k_{\ell})$ by means of EMS protocol as described in
Section~\ref{sec:EMS}. We will prove  that if an attacker $\A$ only knows
one bit, let us say $k_i$ with $i \in \{1,\dots{},\ell\}$, then $\A$ is able to
recompute any bit $k_j$ with $j \ne i$.   This proves that
EMS protocol does not satisfy the resiliency property unlike what is claimed
by the authors in \cite{EMS15}.
But before presenting our attack, we need an important lemma.

\begin{lemma}\cite[Lemma 5.1]{BGH07}\label{lem:2} 
Let $A$, $B_1$ and $B_2$ be elements from $\Z_N$, and for each $i$ 
 in $\{ 1, 2 \}$ let $(x_i,y_i)$ be a solution to
\[
Ax^2+B_i y^2 = 1 \mod N.
 \]
If  $A x_1 x_2 + 1$ belongs to $\Zx_N$ then $(x_3,y_3)$ 
with
$x_3 =\frac{x_1+x_2}{1 + A x_1 x_2}$
and
$y_3=\frac{y_1y_2}{1 + A x_ 1x_2}$
is solution to
\[
A x^2 + B_1 B_2 y^2 = 1 \mod N.
\]
\end{lemma}

We now describe how an adversary $\A$ can break the EMS protocol if $\A$ only knows
$k_i$ for some $i \in \{1,\dots{},\ell\}$ from a shared key $(k_1,\dots{},k_\ell)$.
For the sake of simplicity, we will only describe how $\A$ can recover $k_{i+1}$
from $k_i$ and data publicly exchanged by $P_1$ and $P_2$. By induction, the attack
can be generalized to any bit $k_j$.

\medskip

Firstly, $\A$ solves \eqref{eq:EMS} for $i$ and $i+1$  to get $(x_i,y_i)$ and $(x_{i+1},y_{i+1})$ such that:
\[
\left\{
\begin{array}{lclcl}
\Tun S_1^{2i+1} x_i^2 &+& \Tde S_2^{2i+1}y_i^2 &=& 1  \mod N\\
\\ \notag
\Tun S_1^{2i+3} x_{i+1}^2 &+& \Tde S_2^{2i+3}y_{i+1}^2 &=& 1 \mod N.
\end{array}
\right.
\]
As explained in Section~\ref{sec:SolvingQuad},
$\A$ gets the \emph{same} solutions to these equations as $P_1$ and $P_2$ would have
during the protocol. Furthermore, $\A$ knows $(S_1 x_{i+1},y_{i+1})$ which is
a solution to the following equation:
\begin{equation*}
\Tun S_1^{2i+1} \left(S_1x_{i+1}\right)^2 + \Tde S_2^{2i+3}y_{i+1}^2 = 1 \mod N.
\end{equation*}
From solutions $(x_i,y_i)$ and
$(S_1 x_{i+1},y_{i+1})$, the adversary $\A$, by using Lemma~\ref{lem:2},
derives $(x_*,y_*)$ that is solution to the equation
\begin{equation} \label{eq:x*y*}
\Tun S_1^{2i+1} x_*^2 + \Tde^2 S_2^{4i+4}y_*^2 = 1 \mod N
\end{equation}
where
\[
x_* = \frac{x_i + S_1 x_{i+1}}{1 + \Tun S_1^{2i+2} x_ix_{i+1}}
~~~
\text{and}
~~~
y_* = \frac{y_i y_{i+1}}{1 + \Tun S_1^{2i+2} x_ix_{i+1}}.
\]

The next lemma proves that $k_i$ and $k_{i+1}$ are related and an adversary
can easily compute $k_{i+1}$ from $k_i$ and $y_*$ and the public data exchanged
between $P_1$ and $P_2$.

\begin{lemma}
Let $\left(x_*,y_* \right)$ be \emph{the} solution to \eqref{eq:x*y*}. We then have the equality:
\begin{equation*}
k_{i+1}=
k_i \cdot
\left(\frac{1 + \Tun S_1^{2i+2} x_ix_{i+1}}{N}\right)\cdot\left(\frac{2 + 2y_* \Tde S_2^{2i+2}}{N}\right)
\end{equation*}
\end{lemma}

\begin{proof}
We start by observing that
\begin{eqnarray*}
1 + x_* \sqrt{\Tun S_1^{2i+1}}
&=&  1 +  x_* \sqrt{\Tun} s_1^{2i+1} \\
&=& 1 + \frac{\left(x_i + S_1 x_{i+1}\right)}{1 + \Tun S_1^{2i+2} x_i x_{i+1}} \sqrt{\Tun} s_1^{2i+1}\\
&=& \frac{1 + \Tun S_1^{2i+2} x_i x_{i+1} + x_i \sqrt{\Tun} s_1^{2i+1}
+ x_{i+1} \sqrt{\Tun} s_1^{2i+3}}
{1 + \Tun S_1^{2i+2} x_ix_{i+1}}\\
&=& \frac{\left(1 + x_i \sqrt{\Tun} s_1^{2i+1} \right)
\left(1 + x_{i+1} \sqrt{\Tun} s_1^{2i+3} \right)}{1 + \Tun S_1^{2i+2} x_ix_{i+1}}
\end{eqnarray*}
Since $k_i = \left(\frac{1 + x_i \sqrt{\Tun} s_1^{2i+1}}{N}\right)$ and
$k_{i+1} = \left(\frac{1 + x_{i+1} \sqrt{\Tun} s_1^{2i+3}}{N}\right)$,
this implies in particular that:
\[
k_{i+1} =
k_i\cdot
\left(\frac{1 + \Tun S_1^{2i+2} x_ix_{i+1}}{N}\right)\cdot\left(\frac{1 + x_* \sqrt{\Tun S_1^{2i+1}}}{N}\right)
\]
Since $(x_*,y_*)$ is solution to \eqref{eq:x*y*} 
and $\Tun S_1^{2i+1} \in \QR(N)$
then 
by Lemma~\ref{lem:Jacobi_xy} we also have that:

\begin{equation*}
\left(\frac{2 + 2y_* \Tde S_2^{2i+2}}{N}\right)
=
\left(\frac{1 + x_* \sqrt{\Tun S_1^{2i+1}}}{N}\right)
\end{equation*}
which terminates the proof of the lemma.
\quad \end{proof}

This attack is as efficient as the scheme since it only requires to compute the Jacobi
symbol and the solving of \eqref{eq:BGH}.
The computation of the Jacobi symbol can be performed \cite{vonGG13} 
in $O\left( \log N \mathsf{M}\left( \log N\right) \right)$ 
operations where $\mathsf{M}(\lambda)$ is the cost of the multiplication 
of two integers of 
size $\lambda$ bits (for large integers 
$\mathsf{M}(\lambda) = \lambda \log \lambda \log \log \lambda$). 
The equation \eqref{eq:BGH} can be solved with $O\left(\log^4 N \right)$ operations (\cite{BGH07}). The total cost of the attack is therefore 
$O\left(\log^4 N \right)$ operations.

\section{Discussion on a reparation} \label{sec:discussion}

Our work raises also the question of whether the EMS protocol can be repaired. 
Our attack exploits the fact that the shared secret bits are related (see Lemma 3).
One possible reparation would be to generate $\ell$ \emph{independent} 
identity values $R_{i,1},\dots{},R_{i,\ell}$ for a party $P_i$. 
For instance, one solution is to set $R_j = \hash(\id,j)$ 
for all $j \in \{1,\dots{},\ell\}$, and the authority creates the secret key as 
$\sk = ( (a_j)_{1\leq j \leq \ell}, (r_j)_{1\leq j \leq \ell} )$ with 
$r_j = \sqrt{\mu^{a_j}R_j}$.
Next, when two parties $P_1$ and $P_2$ wish to authenticate, 
they just have to solve 
the equations:
\begin{equation}
R_{1,i} x_{i,j}^2 + R_{2,j} y_{i,j}^2 = 1 \mod N.
\end{equation}
The secret shared bit associated to this equation is now 
$k_{i,j} = 
\left(\frac{1 + x_{i,j} \sqrt{R_{1,i}}}{N}\right)
=
\left(\frac{2 + 2y_{i,j} \sqrt{R_{2,j}}}{N}\right).
$
Hence, they do not require anymore the values $\mu^{\alpha_1}$ and 
$\mu^{\alpha_2}$ which 
introduced the weaknesses (in particular for the impersonation attack).
Unfortunately, this new protocol can (only) produce $\ell^2$ secret bits, 
and the secret and public keys becomes $\ell$ times larger, which leads to 
an inefficient protocol for realistic applications.

\section{Conclusion} \label{sec:conclusion}

In this paper, we have studied the security of the IBAKE protocol in introduced in \cite{EMS15}.
The authors claimed that their protocol is provably secure when during the key exchange session 
some secret bits are leaked.

\medskip

We showed that this protocol has two major weaknesses. First, it is vulnerable to a simple 
man-in-the-middle attack.  Secondly, we propose an efficient attack where an adversary can 
easily compute any bit of a shared key if just one secret bit is known, which 
contradicts authors' claim.

\section*{Acknowledgments}

A. Otmani is funded by ANR grant ANR-15-CE39-0013-01 MANTA.

\bibliographystyle{plain}

\begin{thebibliography}{100}
\bibitem{BX04}
Dan Boneh and Xavier Boyen.
\newblock Efficient selective-id secure identity-based encryption without
  random oracles.
\newblock In Christian Cachin and Jan~L. Camenisch, editors, {\em Advances in
  Cryptology - EUROCRYPT 2004: International Conference on the Theory and
  Applications of Cryptographic Techniques, Interlaken, Switzerland, May 2-6,
  2004. Proceedings}, pages 223--238, Berlin, Heidelberg, 2004. Springer Berlin
  Heidelberg.

\bibitem{BF01}
Dan Boneh and Matt Franklin.
\newblock Identity-based encryption from the weil pairing.
\newblock In Joe Kilian, editor, {\em Advances in Cryptology --- CRYPTO 2001:
  21st Annual International Cryptology Conference, Santa Barbara, California,
  USA, August 19--23, 2001 Proceedings}, pages 213--229, Berlin, Heidelberg,
  2001. Springer Berlin Heidelberg.

\bibitem{BGH07}
Dan Boneh, Craig Gentry, and Michael Hamburg.
\newblock Space-efficient identity based encryption without pairings.
\newblock {\em 2007 48th Annual IEEE Symposium on Foundations of Computer
  Science}, 00:647--657, 2007.

\bibitem{C01}
Clifford Cocks.
\newblock An identity based encryption scheme based on quadratic residues.
\newblock In {\em Proceedings of the 8th IMA International Conference on
  Cryptography and Coding}, pages 360--363, London, UK, UK, 2001.
  Springer-Verlag.

\bibitem{CR03}
John Cremona and David Rusin.
\newblock Efficient solution of rational conics.
\newblock {\em Mathematics of Computation}, 72(243):1417--1441, 2003.

\bibitem{DH76}
W.~Diffie and M.~Hellman.
\newblock New directions in cryptography.
\newblock {\em IEEE Trans. Inf. Theor.}, 22(6):644--654, September 1976.

\bibitem{EMS15}
Ibrahim Elashry, Yi~Mu, and Willy Susilo.
\newblock A resilient identity-based authenticated key exchange protocol.
\newblock {\em Security and Communication Networks}, 8(13):2279--2290, 2015.
\newblock sec.1172.

\bibitem{GPV08}
Craig Gentry, Chris Peikert, and Vinod Vaikuntanathan.
\newblock Trapdoors for hard lattices and new cryptographic constructions.
\newblock In {\em Proceedings of the Fortieth Annual ACM Symposium on Theory of
  Computing}, STOC '08, pages 197--206, New York, NY, USA, 2008. ACM.

\bibitem{JB09}
Mahabir~Prasad Jhanwar and Rana Barua.
\newblock A variant of boneh-gentry-hamburg's pairing-free identity based
  encryption scheme.
\newblock In Moti Yung, Peng Liu, and Dongdai Lin, editors, {\em Information
  Security and Cryptology: 4th International Conference, Inscrypt 2008,
  Beijing, China, December 14-17, 2008, Revised Selected Papers}, pages
  314--331, Berlin, Heidelberg, 2009. Springer Berlin Heidelberg.

\bibitem{J04}
Antoine Joux.
\newblock A one round protocol for tripartite diffie--hellman.
\newblock {\em Journal of Cryptology}, 17(4):263--276, 2004.

\bibitem{LI11}
Xiong Li, Jian-Wei Niu, Jian Ma, Wen-Dong Wang, and Cheng-Lian Liu.
\newblock Cryptanalysis and improvement of a biometrics-based remote user
  authentication scheme using smart cards.
\newblock {\em Journal of Network and Computer Applications}, 34(1):73 -- 79,
  2011.

\bibitem{LI13}
Xiong Li, Jianwei Niu, Muhammad~Khurram Khan, and Junguo Liao.
\newblock An enhanced smart card based remote user password authentication
  scheme.
\newblock {\em Journal of Network and Computer Applications}, 36(5):1365 --
  1371, 2013.

\bibitem{MB05}
Noel McCullagh and Paulo S. L.~M. Barreto.
\newblock A new two-party identity-based authenticated key agreement.
\newblock In Alfred Menezes, editor, {\em Topics in Cryptology -- CT-RSA 2005:
  The Cryptographers' Track at the RSA Conference 2005, San Francisco, CA, USA,
  February 14-18, 2005. Proceedings}, pages 262--274, Berlin, Heidelberg, 2005.
  Springer Berlin Heidelberg.

\bibitem{S85}
Adi Shamir.
\newblock Identity-based cryptosystems and signature schemes.
\newblock In {\em Proceedings of CRYPTO 84 on Advances in Cryptology}, volume
  196, pages 47--53, New York, NY, USA, 1985. Springer-Verlag New York, Inc.

\bibitem{S01}
N.~P. Smart.
\newblock An identity based authenticated key agreement protocol based on the
  weil pairing.
\newblock {\em Electronics Letters}, 38:630--632, 2001.

\bibitem{TITN16}
Ferucio~Lauren{\c{T}}iu {\c{T}}iplea, Sorin Iftene, George Te{\c{s}}eleanu, and
  Anca-Maria Nica.
\newblock Security of identity-based encryption schemes from quadratic
  residues.
\newblock In Ion Bica and Reza Reyhanitabar, editors, {\em Innovative Security
  Solutions for Information Technology and Communications: 9th International
  Conference, SECITC 2016, Bucharest, Romania, June 9-10, 2016, Revised
  Selected Papers}, pages 63--77, Cham, 2016. Springer International
  Publishing.

\bibitem{vonGG13}
J.~von~zur Gathen and J.~Gerhard.
\newblock {\em Modern Computer Algebra}.
\newblock Cambridge University Press, 2013.

\bibitem{W05}
Brent Waters.
\newblock Efficient identity-based encryption without random oracles.
\newblock In Ronald Cramer, editor, {\em Advances in Cryptology -- EUROCRYPT
  2005: 24th Annual International Conference on the Theory and Applications of
  Cryptographic Techniques, Aarhus, Denmark, May 22-26, 2005. Proceedings},
  pages 114--127, Berlin, Heidelberg, 2005. Springer Berlin Heidelberg.
\end{thebibliography}

\end{document}